\documentclass[a4paper]{llncs}
 
\usepackage{algorithm}
\usepackage[noend]{algorithmic}
\usepackage{amsmath}
\usepackage{amssymb}
\usepackage{graphicx}
\usepackage{dsfont}
\usepackage{amsfonts}
\usepackage{setspace}
\usepackage[shortlabels]{enumitem}
\newtheorem{observation}[theorem]{Observation}

\makeatletter

\newcommand{\ALOOP}[1]{\ALC@it\algorithmicloop\ #1%
  \begin{ALC@loop}}
\newcommand{\ENDALOOP}{\end{ALC@loop}\ALC@it\algorithmicendloop}

\makeatother

\renewcommand{\epsilon}{\varepsilon}

\newcommand{\old}[1]{{}}
\newcommand{\versionA}[1]{{}}

\let\doendproof\endproof
\renewcommand\endproof{~\hfill\qed\doendproof}

\title{$\delta$-Greedy $t$-spanner
}
\author{Gali Bar-On \and Paz Carmi }
\institute{Department of Computer Science,\\ Ben-Gurion University of the Negev, Israel}

\begin{document}
\maketitle


\begin{abstract}
We introduce a new geometric spanner, $\delta$-\emph{Greedy}, whose construction is based on a generalization of the 
known \emph{Path-Greedy} and \emph{Gap-Greedy} spanners. 
The $\delta$-Greedy spanner combines the most desirable properties of geometric spanners both in theory and in practice.
More specifically, it has the same theoretical and practical properties as the Path-Greedy spanner: a natural definition, small degree, linear number of edges, low weight, and strong $(1+\varepsilon)$-spanner for every 
$\varepsilon>0$.  
The  $\delta$-Greedy algorithm is an improvement over the Path-Greedy algorithm with respect to the number of shortest path queries and hence with respect to its construction time. 
We show how to construct such a spanner for a set of $n$ points in the plane in $O(n^2 \log n)$ time. 

The  $\delta$-Greedy spanner has an additional parameter, $\delta$, which indicates how close it is to the Path-Greedy spanner on the account of the number of shortest path queries. 
For $\delta = t$ the output spanner is identical to the Path-Greedy spanner, 
while the number of shortest path queries is, in practice, linear. 

Finally, we show that for a set of $n$ points placed independently at random in a unit square the expected construction time of the $\delta$-Greedy algorithm is $O(n \log n)$. 
Our analysis indicates that the $\delta$-Greedy spanner gives the best results among the known 
spanners of expected $O(n \log n)$ time for random point sets. 
Moreover, the analysis implies that by setting $\delta = t$,  the $\delta$-Greedy algorithm provides a spanner identical to the Path-Greedy spanner in expected $O(n \log n)$ time.
\end{abstract}


\section{Introduction}\label{sec:Intro}
Given a set $P$ of points in the plane, a Euclidean $t$-spanner for $P$ is an undirected graph $G$, 
where  there is a $t$-spanning path in $G$ between any two points in $P$. 
A path between points $p$ and $q$  is a $t$-spanning path if its length is at most $t$ times the Euclidean distance 
between $p$ and $q$ (i.e., $t|pq|$).

The most known algorithm for computing $t$-spanner is probably the \emph{Path-Greedy} spanner.
Given a set $P$ of $n$ points in the plane, the Path-Greedy spanner algorithm creates a $t$-spanner for $P$ as follows. 
It starts with a graph $G$ having a vertex set $P$, an empty edge set $E$ and $ {n \choose 2} $ 
pairs of distinct points sorted in a non-decreasing order of their distances. Then, it adds an edge between $p$ and $q$ 
to the set $E$ if the length of the shortest path between $p$ and $q$ in $G$ is more than $t|pq|$, see Algorithm~\ref{alg:pathGreedy} for more details.
It has been shown in~\cite{Chandra,Chandra94,Das1,DasHN93,GudmundssonLN02,Soares1994} that for every set of points, 
the Path-Greedy spanner has $O(n)$ edges, a bounded degree and
total weight $O(wt(MST(P)))$, where $wt(MST(P))$ is the weight of a minimum spanning tree of $P$.   
The main weakness of the Path-Greedy algorithm is its time complexity -- the naive implementation of the Path-Greedy algorithm runs in near-cubic time. By performing $n \choose 2$ shortest path queries, where each query uses Dijkstra's shortest path algorithm, 
the time complexity of the entire algorithm reaches $O(n^3 \log n)$, where $n$ is the number of points in $P$.
Therefore,  researchers in this field have been trying to improve the Path-Greedy algorithm time complexity. 
For example, the \emph{Approximate-Greedy} algorithm generates a graph with the same theoretical properties as the Path-Greedy spanner in $O(n \log n)$ time~\cite{DBLP97,DBLP02}.
However, in practice there is no correlation between the expected and the unsatisfactory resulting spanner 
as shown in~\cite{DBLP07,FarshiG09}. Moreover, the algorithm is complicated and difficult to implement.

Another attempt to build a $t$-spanner more efficiently is introduced in~\cite{FarshiG05,DBLP07}.
This algorithm uses a matrix to store the length of the shortest path between every two points. 
For each pair of points, it first checks the matrix to see if there is a $t$-spanning path between these points. 
In case the entry in the matrix for this pair indicates that there is no $t$-spanning path, it performs a shortest 
path query and updates the matrix.
The authors in~\cite{DBLP07} have conjectured that the number of performed shortest path queries is linear.
This has been shown to be wrong in~\cite{BCFMS08}, as the number of shortest path queries may be quadratic. 
In addition, Bose et al.~\cite{BCFMS08} have shown how to compute the Path-Greedy spanner in $O(n^2\log n)$ time. 
The main idea of their algorithm is to compute a partial shortest path and then extend it when needed. 
However, the drawback of this algorithm is that it is complex and difficult to implement. 
In~\cite{AlewijnseBBB15}, Alewijnse et al. compute the Path-Greedy spanner using linear space in $O(n^2\log^2n)$ time by utilizing the 
Path-Greedy properties with respect to the Well Separated Pair Decomposition (WSPD).
In~\cite{Alewijnse2016}, Alewijnse et al. compute a $t$-spanner in $O(n \log^2 n\log^2\log n)$ expected time by using bucketing for short edges and by using WSPD for  long edges. Their algorithm is based on the  assumption that the Path-Greedy spanner consists of mostly short edges.

\begin{algorithm}[t]
\caption{Path-Greedy$(P,t)$}\label{alg:pathGreedy}
 \begin{algorithmic}[1]
    \REQUIRE A set $P$ of points in the plane and a constant $t > 1$ 
    \ENSURE A $t$-spanner $G(V,E)$ for $P$  
    \STATE sort the $n \choose 2$ pairs of distinct points in non-decreasing order of their distances and  
		       store them in list $L$ 
    \STATE $E \longleftarrow \emptyset$
		\FOR {$ (p,q) \in L$   \ \  consider pairs in increasing order }
		   \STATE $ \pi \longleftarrow$ length of the shortest path in $G$ between $p$ and $q$
			 \IF {$ \pi > t |pq|$} 
				   \STATE $E:=E\cup|pq|$
			\ENDIF
		\ENDFOR	 
	\RETURN $G=(P,E)$
\end{algorithmic}
\end{algorithm}  

Additional effort has been put in developing algorithms for computing $t$-spanner graphs, such as $\theta$-Graph algorithm~\cite{Clarkson87,Kei88}, Sink spanner, Skip-List spanner~\cite{AryaMS94}, and WSPD-based spanners~\cite{Callahan93,CallahanK92}. 
However, none of these algorithms produces a $t$-spanner as good as the Path-Greedy spanner in all aspects: size, weight and maximum degree, see~\cite{DBLP07,FarshiG09}. 

Therefore, our goal is to develop a simple and efficient algorithm  
that achieves both the theoretical and practical properties of the Path-Greedy spanner. 
In this paper we introduce the $\delta$-Greedy algorithm that constructs such a spanner for a set of $n$ points in the plane in $O(n^2 \log n)$ time.
Moreover, we show that for a set of $n$ points placed independently at random in a unit square the expected running time of the $\delta$-Greedy algorithm is $O(n \log n)$.

\section{ $\delta$-Greedy}\label{sec:delta-Greedy}
In this section we describe the $\delta$-Greedy algorithm (Section~\ref{sec:algDes}) for a given set $P$ of points in the plane, and two real numbers $t$ and $\delta$, such that  $1 < \delta \leq t$. Then, in Section~\ref{subSec:SR} we prove that the resulting graph is indeed a $t$-spanner with bounded degree.
Throughout this section we assume that $\delta < t$ (for example, $\delta = t^{\frac{4}{5}}$ or $\delta = \frac{1 + 4t}{5}$), except in Lemma~\ref{lemma:equal}, where we consider the case that $\delta=t$.
 
%
\subsection{Algorithm description}\label{sec:algDes}
For each point $p \in P$ we maintain a collection of cones $C_p$ with the property that for each point $q \in P$ that lies in $C_p$ there is a $t$-spanning path between $p$ and $q$ in the current graph.
The main idea of the $\delta$-Greedy algorithm is to ensure that two cones of a constant angle with apexes at $p$ and $q$ are added to $C_p$ and to $C_q$,  respectively, each time the algorithm runs a shortest path query between points $p$ and $q$.
%
%

The algorithm starts with a graph $G$ having a vertex set $P$, an empty edge set, and an initially empty
collection of cones $C_p$ for each point $p \in P$. 
The algorithm considers all pairs of distinct points of $P$ in a non-decreasing order of their distances.
If $p \in C_q$ or $q \in C_p$, then there is already a $t$-spanning path that connects $p$ and $q$ in $G$, and
there is no need to check this pair.
Otherwise, let $d$ be the length of the shortest path that connects $p$ and $q$ in $G$ divided by $|pq|$.
Let $c_p(\theta,q)$ denote the cone with apex at $p$ of angle $\theta$, such that the ray  $\stackrel{\rightarrow}{pq}$ is its bisector.
The decision whether to add the edge $(p, q)$ to the edge set of $G$ is made according to the value of $d$.
If $d > \delta$, then we add the edge $(p,q)$ to $G$, a cone $c_p (2 \theta,q)$ to $C_p$, and a cone $c_q (2 \theta,p)$ to $C_q$, 
where $\theta =  \frac{\Pi}{4} - \arcsin(\frac{1}{\sqrt 2 \cdot t})$.
If $d  \leq \delta$, then we do not add this edge to $G$, however, 
we add a cone $c_p (2 \theta,q)$ to $C_p$ and a cone $c_q (2 \theta,p)$ to $C_q$, where 
$\theta =  \frac{\Pi}{4} - \arcsin(\frac{d}{\sqrt 2 \cdot t})$.

\begin{figure}[b] 
    \centering
        \includegraphics[width=0.8\textwidth]{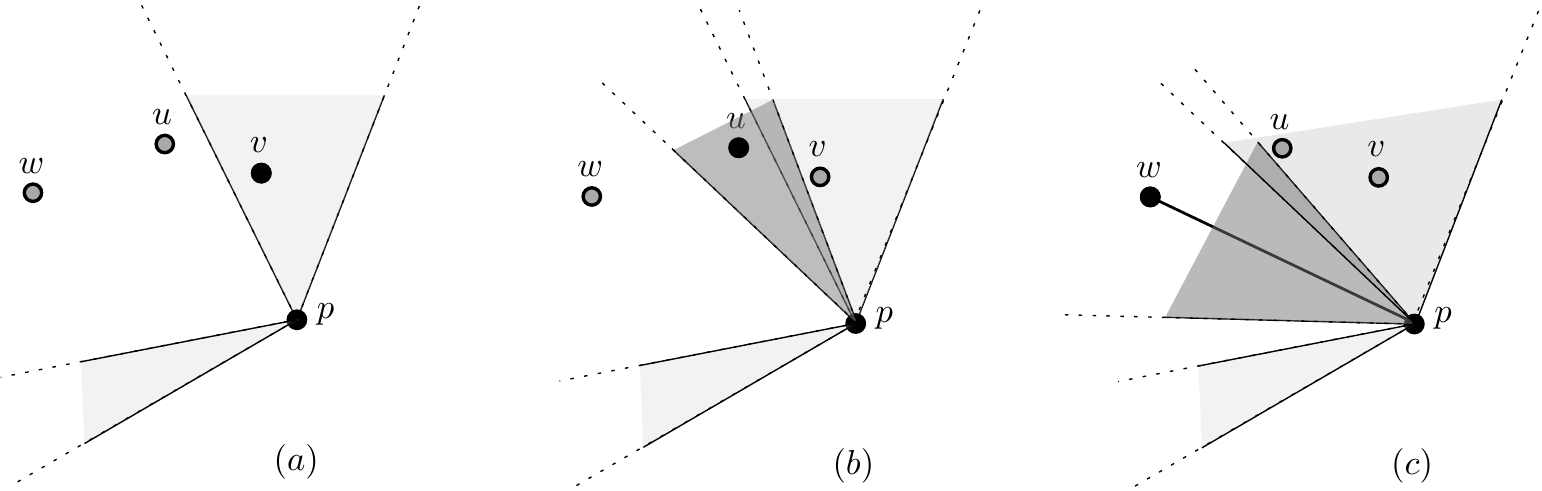}
    \caption{The three scenarios of the $\delta$-Greedy algorithm. (a) $v \in C_p$; (b) $u \notin C_p$ and $d \leq \delta$; 
		           (c)~$w \notin C_p$ and $d > \delta$. }
    \label{fig:C_p}
\end{figure}

In Algorithm~\ref{alg:deltaGreedy}, we give the pseudo-code description of the $\delta$-Greedy algorithm.
In Figure~\ref{fig:C_p}, we illustrate a cone collection $C_p$ of a point $p$ and how it is modified during the three scenarios of the algorithm.  The figure contains the point $p$, its collection $C_p$ colored in gray, and three points $v$, $u$, and $w$, such
that $|pv| < |pu| < |pw|$. Point $v$ lies in $C_p$ representing the first case, where the algorithm does not change the spanner and proceeds to the next pair without performing a shortest path query. The algorithm runs a shortest path query between $p$ and $u$, 
since $u \notin C_p$  (for the purpose of illustration assume $p \notin C_u$).  
Figure~\ref{fig:C_p}(b) describes the second case of the algorithm, where the length of the shortest path between $p$ and $u$ is at most 
$\delta|pu|$. In this case the algorithm adds a cone to $C_p$ without updating the spanner.
Figure~\ref{fig:C_p}(c) describes the third case of the algorithm, where the length of the shortest path between $p$ and $w$ is more than 
$\delta|pw|$. In this case the algorithm adds a cone to $C_p$ and the edge $(p,w)$ to the spanner.


\begin{algorithm}
\caption{$\delta$-Greedy}\label{alg:deltaGreedy}
  \begin{algorithmic}[1] 
    \REQUIRE A set $P$ of points in the plane and two real numbers $t$ and $\delta$ s.t. $1 < \delta \leq t$
    \ENSURE A $t$-spanner for $P$
	\STATE sort the $n \choose 2$ pairs of distinct points in non-decreasing order of their distances 
	  (breaking ties arbitrarily) and store them in list $L$ 
   \STATE $E \longleftarrow \emptyset$  \hspace{6.4cm}  
	                                 /* E is the edge set */
	\STATE $C_p \longleftarrow \emptyset \ \  \forall p \in P$ \hspace{2.48cm} 
	                                /* $C_p$ is set of cones with apex at $p$ */ 
	\STATE $G \longleftarrow (P,E)$	 \hspace{4.01cm} 
	                             /* G  is the resulting $t$-spanner  */
		\FOR{$(p,q) \in L$  \ \  consider pairs in increasing order}\label{alg:edgeIteration}
	 \IF {$(p \notin C_q)$  and $(q \notin C_p)$}  
			 \STATE\label{Alg:shortestPath} $ d \longleftarrow $ 
			length of the shortest path in $G$ between $p$ and $q$ divided $|pq|$
			\IF {$d > \delta $ } 
				\STATE\label{StepAddingEdges} $E \longleftarrow E \cup \{ (p,q) \}$
				\STATE $ d \longleftarrow 1 $
			\ENDIF
			\STATE  $\theta  \longleftarrow \frac{\Pi}{4} - \arcsin(\frac{d}{\sqrt 2 \cdot t})$  \hspace{3.84cm}  
			                      /* $\frac{1}{\cos \theta - \sin \theta} =  \frac{t}{d}$ */
			\STATE\label{alg:addConesP} $c_p (2 \theta,q)  \longleftarrow$ cone of angle  $2 \theta$ with apex at $p$
			and bisector  $ \stackrel{\rightarrow}{pq}$ 
      \STATE $c_q (2\theta ,p) \longleftarrow$ cone of angle  $2 \theta$ with apex at $q$ and
			 bisector $ \stackrel{\rightarrow}{qp}$  
 
			\STATE $C_p \longleftarrow C_p \cup c_p (2 \theta,q) $
			\STATE $C_q \longleftarrow C_q \cup c_q (2 \theta,p)$			
	\ENDIF		  	
\ENDFOR	
		\RETURN $G=(P,E)$
	\end{algorithmic}
	\end{algorithm}

\subsection{Algorithm analysis}\label{subSec:SR}
In this section we analyze several properties of the $\delta$-Greedy algorithm, 
including the spanning ratio and the degree of the resulting graph. 

The following lemma is a generalization of Lemma~6.4.1. in~\cite{GiriSmid07}.
\begin{lemma}\label{lemma:theta}
Let $t$ and $\delta$ be real numbers, such that $1  \leq \delta \leq t$. 
Let $p$, $q$, and $r$ be points in the plane, such that
\begin{enumerate}
	\item $p \neq r$,
	\item $ |pr| \leq |pq|$,
	\item  $\frac {1} {\cos \theta - \sin \theta} \leq \frac{t}{\delta}$, where $\theta$ is the angle $\angle rpq$
	       \ $($i.e., $\angle rpq = \theta \leq  \frac{\Pi}{4} - \arcsin(\frac{\delta}{\sqrt 2 \cdot t}) )$.
\end{enumerate}
\vspace{0.2cm}
Then  $\delta|pr|+ t|rq| \leq t|pq|$.
\end{lemma}
\begin{proof}
Let $r'$ be the orthogonal projection of $r$ onto segment $\overline{pq}$. 
Then, $|rr'| = |pr| \sin \theta$, $|pr'| = |pr| \cos \theta$, and $|r'q| = |pq| - |pr'|$. Thus,
$|r'q| = |pq| - |pr| \cos \theta$.  
By triangle inequality  
\begin{align*}  
  |rq| & \leq  |rr'| + |r'q| \\
	     & \leq  |pr| \sin \theta + |pq| - |pr| \cos \theta \\
			 & =    |pq| - |pr|( \cos \theta - \sin \theta).
\end{align*}
\begin{eqnarray*}  
 \text{We have, \ }  \delta|pr|+ t|rq| &\leq& \delta |pr|  + t (|pq| - |pr|( \cos \theta - \sin \theta) )  \\
                       &=&    t|pq| -  t|pr| ( \cos \theta - \sin \theta) + \delta |pr| \\  
                       &\leq& t|pq| -  t|pr| ( \cos \theta - \sin \theta) + t ( \cos \theta - \sin \theta)  |pr| \\
											&\leq& t|pq|.
\end{eqnarray*}

\old{
We have 
\begin{eqnarray*}  
  \delta|pr|+ t|rq| &\leq& \delta |pr|  + t (|pq| - |pr|( \cos \theta - \sin \theta) )  \\
                       &=&    t|pq| -  t|pr| ( \cos \theta - \sin \theta) + \delta |pr| \\  
                       &\leq& t|pq| -  t|pr| ( \cos \theta - \sin \theta) + t ( \cos \theta - \sin \theta)  |pr| \\
											&\leq& t|pq|.
\end{eqnarray*}
} 
\end{proof}
\begin{lemma}\label{lemma:shortest-path}
The number of shortest path queries performed by $\delta$-Greedy algorithm 
for each point  is $O(\frac{1}{t/\delta -1})$.
\end{lemma}
\begin{proof}
Clearly, the number of shortest path queries performed 
for each point is at most $n-1$. Thus, we may assume that $t/\delta > 1 + 1/n$.
Consider a point $p\in P$ and let $(p,q)$ and $(p,r)$ be two pairs of points that $\delta$-Greedy algorithm has run 
shortest path queries for.
Assume w.l.o.g. that the pair $(p,r)$ has been considered before the pair $(p,q)$, i.e., $|rp| \leq |pq|$.
Let $d$ be the length of the path computed by the shortest path query for $(p,r)$ divide by $|pr|$.
If $d \leq \delta$, then the cone added to the collection $C_p$ has an angle of at least  
$\frac{\Pi}{4} - \arcsin(\frac{\delta}{\sqrt 2 \cdot t})$.
Otherwise,  the algorithm adds the edge $(p,r)$ to $G$ and a new cone to the collection of cones $C_p$, 
where the angle of this cone is $\frac{\Pi}{4} - \arcsin(\frac{1}{\sqrt 2 \cdot t})$. 
Thus, after the shortest path query performed for the pair $(p,r)$, the collection $C_p$ contains a cone $c_p (\theta, r)$, 
where $\theta$ is at least $\frac{\Pi}{2} - 2\arcsin(\frac{\delta}{\sqrt 2 \cdot t})$.
The $\delta$-Greedy algorithm performs a shortest path query for $(p,q)$ only if $p \notin C_q$ and $q \notin C_p$. 
Thus, the angle $\angle rpq$ is at least  $\frac{\Pi}{4} - \arcsin(\frac{\delta}{\sqrt 2 \cdot t})$, and we 
have at most $k= \frac{2 \pi}{\theta}$ shortest path queries for a point.

Let us consider the case where $t>1$ and $\frac{t}{\delta}  \rightarrow 1$.
The equation  $\theta = \frac{\Pi}{4} - \arcsin(\frac{\delta}{\sqrt 2 \cdot t})$ implies that 
$\frac {1} {\cos \theta - \sin \theta} = \frac{t}{\delta}$.
Then, we have   $$\theta \rightarrow 0 , \ \frac{t}{\delta} \sim  1 + \theta, \ \text{and} \ \theta \sim \frac{t}{\delta} -1.$$
Thus, we have  $k \sim \frac{2\pi}{ \frac{t}{\delta} -1} = O(\frac{1}{t / \delta -1})$.			
 \end{proof}

\begin{observation}
 For $\delta = t^{\frac{x-1}{x}} $, where $x>1$ is a fixed integer, 
the number of shortest path queries performed by $\delta$-Greedy algorithm for each point  is $O(\frac{x}{t -1})$.
\end {observation} 
\begin{proof}
\old{
As in Lemma~\ref{lemma:shortest-path}, let us consider the case where $t>1$ and $\frac{t}{\delta}  \rightarrow 1$.
Then, we have   $$\theta \rightarrow 0, \ \frac{t}{\delta} \sim  1 + \theta, \ 
                   \frac{t}{t^{(\frac{x-1}{x})}} \sim  1 + \theta, \ t^{( \frac{1}{x})}  \sim  1 + \theta, \ 
									t \sim (1+\theta)^x, \
									  t \sim  1 + x \cdot \theta,	\ \text{and}  \ \theta \sim \frac{t- 1}{x}.$$
}  
As in Lemma~\ref{lemma:shortest-path}, let us consider the case where $t>1$ and $\frac{t}{\delta}  \rightarrow 1$.
Then, we have   
$$\theta \rightarrow 0, \ \ \frac{t}{\delta} \sim  1 + \theta, \ \ 
\frac{t}{t^{(\frac{x-1}{x})}} \sim  1 + \theta,  \ \ t^{( \frac{1}{x})}  \sim  1 + \theta, $$	
$$ t \sim (1+\theta)^x,  \ \   t \sim  1 + x \cdot \theta,	\ \text{and}  \ \theta \sim \frac{t- 1}{x}.$$
Thus, we have  $k \sim \frac{2\pi x}{ t -1} = O(\frac{x}{t  -1})$.			

\end{proof}

\begin{observation} 
The running time of $\delta$-Greedy algorithm is $O(\frac{n^2 \log n}{(t/\delta -1)^2})$.
\end {observation} 
\begin{proof}
First, the algorithm sorts the $n \choose 2$ pairs of distinct points in non-decreasing order of their distances, 
this takes $O(n^2 \log n)$ time.
A shortest path query is done by Dijkstra's shortest path algorithm on a graph with $O(\frac{n}{t/\delta -1})$ edges and takes 
$O(\frac{n}{t/\delta -1} + n \log n)$ time. 
By Lemma~\ref{lemma:shortest-path} each point performs $O(\frac{1}{t/\delta -1})$ shortest path queries. 
Therefore, we have that the running time of $\delta$-Greedy algorithm is $O( (\frac{n}{t/\delta -1})^2 \log n )$.
\end{proof}

\begin{observation} 
\label{lemma:cone}
The number of cones that each point has in its collection along the algorithm is constant depending on $t$ and 
$\delta$ ($O(\frac{1}{t/\delta -1})$).
\end{observation}
\begin{proof}
As shown in Lemma~\ref{lemma:shortest-path}, the number of shortest path queries for each point is  $O(\frac{1}{t/\delta -1})$. 
The subsequent step of a shortest path query is the addition of two cones, meaning  that for each point $p$ the number of cones in the collection of cones $C_p$ is $O(\frac{1}{t/\delta -1})$.
\end{proof}

\begin{corollary} 
The additional space for each point $p$ for the collection $C_p$ is constant.
\end {corollary}

\begin{lemma}\label{lemma:Spanner}
The output graph $G=(P,E)$ of $\delta$-Greedy algorithm (Algorithm~\ref{alg:deltaGreedy}) is a $t$-spanner for $P$ 
(for $1< \delta < t$).
\end{lemma}
\begin{proof}
Let $G=(P,E)$ be the output graph of the $\delta$-Greedy algorithm. 
To prove that $G$ is a $t$-spanner for $P$ we show that for every pair $(p,q) \in P$, 
there exists a $t$-spanning path between them in $G$. 
We prove the above statement by induction on the rank of the distance $|pq|$, 
i.e., the place of $(p,q)$ in a non-decreasing distances order of all pairs of points in $P$.

\noindent
\textbf{Base case:} Let $(p, q) $ be the first pair in the ordered list (i.e., the closest pair). 
The edge $(p,q)$ is added to $E$ during the first iteration of the loop in step~\ref{StepAddingEdges} of 
Algorithm~\ref{alg:deltaGreedy}, and thus there is a $t$-spanning path between $p$ and $q$ in $G$.

\noindent
\textbf{Induction hypothesis:} 
For every pair $(r,s) $ that appears before the pair $(p,q)$ in the ordered list, 
there is a $t$-spanning path between $r$ and $s$ in $G$.

\noindent
\textbf{The inductive step:} Consider the pair $(p,q)$. We prove that there is a $t$-spanning path between $p$ and $q$ in $G$. 
If $p \notin C_q$ and $q \notin C_p$, we check whether there is a $\delta$-spanning path in $G$ between $p$ and $q$. 
If there is a path which length is at most $\delta |pq|$, 
then  $ \delta|pq| \leq t|pq|$, meaning there is a $t$-spanning path  between $p$ and $q$ in $G$. 
If there is no path of length of at most $\delta|pq|$, we add the edge $( p,q)$ to $G$, which forms a $t$-spanning path. 

Consider that $p \in C_q$ or $q \in C_p$, and assume w.l.o.g. that $q \in C_p$. 
Let $(p,r)$ be the edge handled in Step~\ref{alg:edgeIteration} in Algorithm~\ref{alg:deltaGreedy} 
when the cone containing $q$ has been added to $C_p$ (Step~\ref{alg:addConesP} in Algorithm~\ref{alg:deltaGreedy}). 
Notice that $|pr| \leq |pq|$.
Step~\ref{Alg:shortestPath} of Algorithm~\ref{alg:deltaGreedy} has computed the value $d$ for the pair $(p,r)$.
In the algorithm there are two scenarios depending on the value of $d$.

The first scenario is when $d > \delta$, then the algorithm has added the edge $(p,r)$ to $G$ and a cone $c_p (\theta,r)$ to $C_p$, where  
$\theta = 2(\frac{\Pi}{4} - \arcsin(\frac{1}{\sqrt 2 \cdot t}))$. 
Thus, the angle between $(p,q)$ and $(p, r)$  is less than $\theta /2$.
Hence, $|rq| < |pq|$ and by the induction hypothesis there is a $t$-spanning path between $r$ and $q$.
Consider the shortest path between $p$ and $q$ that goes through the edge $(p,r)$. The length of this path is at most $|pr| +  t|rq|$.
By Lemma~\ref{lemma:theta}, we have $|pr|+ t|rq| \leq \delta |pr|+ t|rq| \leq t|pq|$ for $\delta = 1$.
Therefore, we have a $t$-spanning path between $p$ and $q$. 
 
The second scenario is when $d \leq \delta$, then the algorithm has added a cone $c_p (\theta,r)$ to $C_p$, where  
$\theta = 2(\frac{\Pi}{4} - \arcsin(\frac{d}{\sqrt 2 \cdot t}))$. 
Thus, the angle between $(p,q)$ and $(p, r)$  is less than $\theta / 2$.
Hence, $|rq| < |pq|$ and by the induction hypothesis there is a $t$-spanning path between $r$ and $q$.
Consider the shortest path between $p$ and $q$ that goes through $r$. The length of this path is at most 
$d|pr| +  t|rq|$.
By Lemma~\ref{lemma:theta}, 
we have $d|pr|+ t|rq| \leq  t|pq|$.
Therefore, we have a t-spanning path between $p$ and $q$. 
\end{proof}


\begin{theorem}
The $\delta$-Greedy algorithm computes a $t$-spanner for a set of points $P$ with the same properties as the 
Path-Greedy $t$-spanner, such as degree and weight, in $O( (\frac{n}{t/\delta -1})^2 \log n )$ time.  
\end{theorem}
\begin{proof}
Clearly, the degree of the $\delta$-Greedy is at most the degree of the Path-Greedy $\delta$-spanner.
The edges of the $\delta$-Greedy spanner satisfy the $\delta$-leap frog property, thus, the weight of 
the $\delta$-Greedy is as  Path-Greedy $t$-spanner. Hence, we can pick $\delta$ close to $t$, such that 
we will have the required bounds.   
\end{proof}

\begin{lemma}\label{lemma:equal}
 If $t=\delta$,  the result of the $\delta$-Greedy algorithm is identical to the result of the Path-Greedy algorithm.
\end{lemma}
\begin{proof}
Assume towards contradiction that for $t=\delta$ the resulting graph of the $\delta$-Greedy algorithm, 
denoted as $G=(P,E)$, differs from the result of the Path-Greedy  algorithm, denoted as $G'=(P,E')$. 
Assuming the same order of the sorted edges,
let $(p,q)$ be the first edge that is different in $G$ and $G'$.
Notice that $\delta$-Greedy algorithm decides to add the edge $(p,q)$ to $G$ when there is no $t$-spanning path between $p$ and $q$ in $G$.
Since until handling the edge $(p,q)$ the graphs $G$ and $G'$ are identical, 
the Path-Greedy algorithm also decides to add the edge $(p,q)$ to $G'$. 
Therefore, the only case we need to consider is $(p,q) \in E'$ and $(p,q) \notin E$.  
The $\delta$-Greedy algorithm does not add an edge $(p,q)$ to $G$ in two scenarios:
\begin{itemize}
	\item  there is a $t$-spanning path between $p$ and $q$ in the current graph $G$ \ -- \	
	which contradicts that the Path-Greedy algorithm adds the edge $(p,q)$ to $G'$;
	\item  $p \in C_q$ or $q \in C_p$ \ -- \ 
	 the $\delta$-Greedy algorithm does not perform a shortest path query between $p$ and $q$.
	Assume w.l.o.g., $q \in C_p$, and let $(p,r)$ be the edge considered in Step~\ref{alg:edgeIteration} 
	in Algorithm~\ref{alg:deltaGreedy} 
when the cone containing $q$ has been added to $C_p$.
The angle of the added cone is $\theta = \frac{\Pi}{2} - 2\arcsin(\frac{d}{\sqrt 2 \cdot t}) $, where $d$ 
is the length of the shortest path between $p$ and $r$ divided $|pr|$. 
Thus, we have $ |pr| \leq |pq|$ and $\frac {1} {\cos \alpha - \sin \alpha} \leq \frac{t}{d}$, where $\alpha \leq \theta$ is the angle 
$\angle rpq $.
Then, by Lemma~\ref{lemma:theta},  $\delta|pr|+ t|rq| \leq t|pq|$, 
and since there is a path from $p$ to $r$ of length at most $\delta |pr|$, 
we have that there is $t$-spanning path between $p$ and $q$ in the current graph.  
This is in contradiction to the assumption that the Path-Greedy algorithm adds the edge $(p,q)$ to $E'$.
\end{itemize}
\end{proof}

\section{$\delta$-Greedy in Expected $O(n \log n)$ Time for Random Set}
\label{subSec:calc-nlogn} 
In this section we show how a small modification in the implementation improves the running time of the $\delta$-Greedy algorithm.
This improvement yields an expected $O(n \log n)$ time for random point sets.
The first modification is to run the shortest path query between points $p$ to $q$ up to $\delta |pq|$.
That is, running Dijkstra’s shortest path algorithm with source $p$ and terminating as soon as the minimum key in the priority queue
is larger than $\delta |pq|$.  


Let $P$ be a set of $n$ points in the plane uniformly distributed in a unit square.
To prove that $\delta$-Greedy algorithm computes a spanner for $P$ in expected $O(n \log n)$ time, 
we need to show that: 

\begin{itemize}
  \item each point runs a constant number of shortest path queries \ -- \ follows from Lemma~\ref{lemma:shortest-path};
  \item the expected number of points visited in each query is constant \ -- \
   The fact that the points are randomly chosen uniformly in the unit square implies that the expected number of points at distance of at most 
	$r$ from point $p$ is $\Theta(r^2 \cdot n)$. 
A shortest path query from a point $p$ to a point $q$ terminates as soon as the minimum key in the priority 
queue exceeds $\delta |pq|$, thus, it is expected to visit $O(n \cdot (\delta|pq|)^2)$ points.

By Lemma~\ref{lemma:shortest-path} the number of shortest path queries performed by the algorithm for a point $p$ is 
$O(\frac{1}{t/\delta -1})$. Each such query defines a cone with apex at $p$ of angle $\Omega(t/\delta -1)$, 
such that no other shortest path query from $p$ will be performed to a point in this cone.
By picking $k=\frac{1}{t/\delta -1}$ and $r= \frac{k}{\sqrt n}$, we have that the expected number of points around each point in a distance of $r$ is $\Theta (k^2) = \Theta ( \frac{1}{(t/\delta -1)^2} )$.

Assume we partition the plane into $k$ equal angle cones with apex at point $p$.  
The probability that there exists a cone that does not contain a point from the set of points of distance 
$\frac{k}{\sqrt n}$ is at most $k \cdot (1- \frac{1}{k})^{k^2}$. 
Let $Q$ be the set of points that $p$ computed a shortest path query to, and let $q \in Q$ be the farthest point in $Q$ from $p$. 
Then, the expected Euclidean distance between $p$ and $q$ is less than $\frac{k}{\sqrt n}$.
Thus, the expected number of points visited by the entire set of shortest path queries from a point is  
$O(\frac{\delta^2 k^2}{t/\delta -1}) = O(\frac{\delta^2}{(t - \delta)^3})$;
	
  \item the next pair to be processed can be obtained in expected $O(\log n)$ time 
without sorting all pairs of distinct points \ -- \ Even-though this is quite straight forward, for completeness we give a short description how this can be done.  
Divide the unit square to $n \times n$ grid cells  of side length $1/n$.
A hash table of size $3n$ is initialized, and for each non-empty grid cell (at most $n$ such cells) we map the points in it to the hash table. 
In addition, we maintain a minimum heap $H_p$ for each point  $p \in P$ (initially empty), 
and one main minimum heap $H$ that contains the top element of each $H_p$. 
Each  heap $H_p$ contains a subset of the pairs that include $p$. 

For each point $p \in P$, all the cells of distance at most $\frac{k}{\sqrt n}$ from $p$ are scanned (using the hash table)
to find all the points in these cells, where $k$ is a parameter that we fix later.  
All the points found in these cells are added to $H_p$ according to their Euclidean distance from $p$.
  
The heap $H$ holds the relevant pairs in an increasing order, therefore the pairs are extracted from the main heap $H$. 
After extracting the minimum pair in $H$ that belongs to a point $p$,  we add to $H$ the next minimum in $H_p$.
To insure the correctness of the heaps, when needed we increase the distance to the scanned cells.
Observe that there may be a pair $(p,q)$ such that $|pq| < |rw|$, where the pair $(r,w)$ is the top pair in $H$. 
This can occur only when the pair $(p,q)$ has not been added to $H_p$ nor $H_q$, and this happens when 
$p \in C_q$ or $q \in C_p$. However, in this case we do not need to consider the pair $(p,q)$.

Notice that the only cells that are not contained in $C_p$ are scanned to add more pairs to $H_p$. 
Thus, points that are in $C_p$ are ignored.

\end{itemize}

Therefore, the total expected running time of the algorithm is $O( \frac{\delta^2}{(t - \delta)^3} n \log n )$. 
Since both $t$ and $t/\delta $ are constants bigger than one, the expected running time of the $\delta$-Greedy algorithm is  $O( n \log n )$.

A very nice outcome of $\delta$-Greedy algorithm and its analysis can be seen when  $\delta$ is  equal to $t$.
Assume that $\delta$-Greedy algorithm (for $\delta = t$) has computed a shortest path query for two points $p$ and $q$ and the length of the received path is $d|pq|$.  
If the probability that  $ t/d > 1 +\varepsilon $ is low 
(e.g, less than 1/2), 
for some constant $\varepsilon >0$,  then $\delta$-Greedy algorithm computes the Path-Greedy 
spanner with linear number of shortest path queries. 
Thus $\delta$-Greedy algorithm computes the Path-Greedy spanner for a point set uniformly distributed in a square in expected $O(n \log n)$ time. 

Not surprisingly our experiments have shown that this probability is indeed low (less than 1/100), since most of the shortest path queries are performed  on pairs of points 
placed close to each other (with respect to Euclidean distance), and thus with a high probability their shortest path contains a constant number of points. 
Moreover, it seems  that for a ``real-life" input this probably  is low. 
Thus, there is a very simple algorithm to compute the Path-Greedy spanner 
in expected $O(n^2 \log n)$ time for real-life inputs, based on the $\delta$-Greedy algorithm

For real-life input we mean that our analysis suggests that in the current computers precision (Memory) one cannot 
create an instance of points set with more than 1000 points, where the Path-Greedy spanner based on the $\delta$-Greedy algorithm has more than $O(n^2 \log n)$ constructing time.

\section{Experimental Results}\label{sec:res}
In this section we discuss the experimental results by considering the properties  of the graphs generated by different algorithms and the number of shortest path queries performed during these algorithms. 
We have implemented the Path-Greedy, $\delta$-Greedy, Gap-Greedy, $\theta$-Graph, Path-Greedy on $\theta$-Graph algorithms.
The Path-Greedy on $\theta$-graph $t$-spanner, first computes a $\theta$-graph $t'$-spanner, where $t'< t$, and then 
runs the Path-Greedy $t/t'$-spanner on this $t'$-spanner. 
The shortest path queries criteria is used for an absolute running time comparison that  is independent of the actual implementation.  
The known theoretical bounds for the algorithms can be found in Table~\ref{table:bounds}. 

 \begin{table}[b]
    \begin{tabular}{| l | l | l | l | l | l | l |}
    \hline
     \textbf{Algorithm} &  \textbf{Edges} & $\frac{\textbf{Weight}}{wt(MST)}$ &  \textbf{Degree} &  \textbf{Time} \\ \hline
    Path-Greedy & $O(\frac{n}{t-1})$ & $O(1)$ & $O(\frac{1}{t-1})$ & $O(n^3 \log n)$ \\ \hline
		Gap-Greedy & $O(\frac{n}{t-1})$ & $O(\log n)$ &  $O(\frac{1}{t-1})$ & $O(n \log^2n)$ \\ \hline
		$\theta$-Graph & $O(\frac{n}{\theta})$ & $O(n)$ & $O(n)$ &  $O(\frac{n}{\theta}\log n)$ \\ \hline
		$\delta$-Greedy & $O(\frac{n}{t/\delta -1})$ &  $O(1)$ & $O(\frac{1}{t/\delta -1})$ &  $O(\frac{1}{t/\delta -1}\cdot n^2\log n)$  \\ \hline
    \end{tabular}
	 \caption{Theoretical bounds of different $t$-spanner algorithms}
	\label{table:bounds} 
\end{table}

The experiments were performed on a set of $8000$ points, with different values of the parameter $\delta$ (between 1 and $t$).
We have chosen to present the parameter $\delta$ for the values 
$t, t^{0.9} $ and $\sqrt t $. This values do not 
have special properties, they where chosen arbitrary to present the behavior of the spanner.

To avoid the effect of specific instances, we have run the algorithms several times and taken the average of the results. 
However, in all the cases the difference between the values is negligible. 
Table~\ref{table:res-1.1}--\ref{table:res-2} show the results of our experiments for 
different values of $t$ and $\delta$. The columns of the weight (divided by $wt(MST)$) and the degree are rounded to integers, 
and the columns of the edges are rounded to one digit after the decimal point (in $k$).

\begin{table}
  \begin{tabular}{| l | l | l | l | l | l |}
		    \hline
   \textbf{Algorithm} & \textbf{$\delta$} &  \textbf{Edges} (in K) &  \textbf{Weight} &  \textbf{Degree} &  \textbf{Shortest path}   \\ 
	              &   &       &  $\overline{wt(MST)}$   &     &  \textbf{queries} (in K)     \\ \hline      
    Path-Greedy	& - &	35.6	&	10	&	17	&	$31996$ \\ \hline
			$\delta$-Greedy &  1.1 &	35.6	&	10	&	17	&	254 \\ \hline
				$\delta$-Greedy &  $1.0896$ & 37.8	&	12	&	18	&	242 \\ \hline
		$\delta$-Greedy &  $1.048$ & 51.6	& 19	&	23	&	204 \\ \hline
		$\theta$-Graph &  - & 376.6	&	454 &	149 & - \\ \hline
		Greedy on $\theta$-Graph & $1.0896$ &37.8	&	12	&	18	&	3005 \\ \hline
		Greedy on $\theta$-Graph & $1.048$ &52	&	19	&	23	&	693 \\ \hline
		Gap-Greedy & - &51.6	&	19	&	23	&	326 \\ \hline
  	\end{tabular}
  \caption{Comparison between several $t$-spanner algorithms for $t=1.1$}
	\label{table:res-1.1} 
\end{table}

\begin{table}\label{table:res-1.5} 
    \begin{tabular}{| l | l | l | l | l | l |}
    \hline
    \textbf{Algorithm} & \textbf{$\delta$} &  \textbf{Edges} (in K) &  \textbf{Weight} &  \textbf{Degree} &  \textbf{Shortest path}   \\ 
	              &   &       &  $\overline{wt(MST)}$   &     &  \textbf{queries} (in K)     \\ \hline      
		Path-Greedy &  - & $15.1$	&	$3$	&	7	& $31996$ \\ \hline
		$\delta$-Greedy &  $1.5$ & $15.1$	&	$3$	&	7	&	$82$ \\ \hline
			$\delta$-Greedy & $1.44$	&	$16$	&	$3$&	8	&	$77$ \\ \hline
		$\delta$-Greedy & $1.224$ & $22.5$	&	$5$	&11	&	$63$ \\ \hline
		$\theta$-Graph &  - & $118.6$	&	$76$	&	53 & - \\ \hline
			Greedy on $\theta$-Graph &  $1.44$ & $16$	&	$3$	&	8	&	$817$ \\ \hline
		Greedy on $\theta$-Graph &  $1.224$ & $22.5$	&	$6$	&	11	&	$198$\\ \hline
		Gap-Greedy &  - & $22.6$	&	$5$	&	11 &	$95$ \\ \hline
    \end{tabular}
		\caption{Comparison between several $t$-spanner algorithms for $t=1.5$}
		\end{table}

 \begin{table}
    \begin{tabular}{ | l | l | l | l | l | l |}
    \hline
    \textbf{Algorithm} & \textbf{$\delta$} &  \textbf{Edges} (in K) &  \textbf{Weight} &  \textbf{Degree} &  \textbf{Shortest path}   \\ 
	              &   &       &  $\overline{wt(MST)}$   &     &  \textbf{queries} (in K)     \\ \hline      
		Path-Greedy &  - & 11.4	&	2	&	5 & $31996$ \\ \hline
		$\delta$-Greedy &  $2$ &11.4&	2	&	5	&	55 \\ \hline
		$\delta$-Greedy &  $1.866$ &11.9&	2	&	5	&	52 \\ \hline
		$\delta$-Greedy &  $1.414$ & 16.3	&	3	&	8	&	44 \\ \hline
		$\theta$-Graph &  - & 85.3	&	48	&	42 & - \\ \hline
		Greedy on $\theta$-Graph & $1.866$ &11.9	&	3	&	6	&	493\\ \hline
		Greedy on $\theta$-Graph & $1.414$ &16.5	&	3	&	8	&	129 \\ \hline
		Gap-Greedy & - & 16	&	3	&	8	&	63 \\ \hline
    \end{tabular}
		\caption{Comparison between several $t$-spanner algorithms for $t=2$}
		\label{table:res-2} 
\end{table}

\subsection{Implementation details}\label{subSec:algo-details}
All the algorithms mentioned above were implemented in Java using JGraphT and JGraph libraries.
The experiments were performed on an Intel ® Xeon® CPU E5-2680 v2 $@$ 2.80 GHz (2 processors) and 128 GB RAM on Windows Server 2012 Standard OS using ECJ for compilation. The sample point sets were generated by Java.util.Random pseudo random number generator.

\subsection{Results analysis}\label{subSec:res-analysis}
The experiments indicate that the $\delta$-Greedy algorithm achieves good results in practice as expected. 
The outcome of the $\delta$-Greedy algorithm for all values of $\delta$, that have been checked, is roughly the same as the results of the Path-Greedy algorithm for all parameters. Compared to other algorithms, 
the $\delta$-Greedy graphs are superior to the graphs produced by the $n^2$-Gap algorithm, and are 
as good as Path-Greedy on $\theta$-Graph, with significantly a lower number of shortest path queries. 
The theoretical complexity of the Path-Greedy on $\theta$-Graph is $O(n^2 \log n)$, same as the $\delta$-Greedy algorithm. 
However in practice the $\delta$-Greedy algorithm computes considerably less shortest path queries. 
Hence, the $\delta$-Greedy algorithm has the same results in weight, size and degree as the Path-Greedy on 
$\theta$-Graph algorithm with better running time. 
 
In addition, Farshi and Gudmundsson in~\cite{FarshiG09} have implemented various spanner algorithms
and shown that the Path-Greedy algorithm for $t=1.1$ and for $t=2$ on random graphs are almost 
identical to ours experimental results in weight, 
size and degree.  Moreover, they have shown that Path-Greedy spanner is the highest quality geometric spanner in terms of edge count, degree and weight.
They have presented  the results for $t=1.1$ and for $t=2$ on random point set with 8000 points.
Moreover, they have shown that the $\theta$-Graph spanner achieves in practice the best results after the 
Path-Greedy spanner for all parameters that have been tested (size, weight and degree) comparing to  
other spanners that they have implemented (such as the Approximate-Greedy, the WSPD-spanner, Skip-list and Sink-Spanner). 
Our experiments show that the $\delta$-Greedy spanner achieves better results than the $\theta$-Graph spanner. 
Thus, combining this with the results in~\cite{FarshiG09}, we conclude  that the $\delta$-spanner achieves the 
highest quality geometric spanner  with respect to 
$\theta$-Graph,  Approximate-Greedy, the WSPD-spanner, Skip-list, Sink-Spanner, and Gap-Greedy spanners. 

The experiments reinforce the analysis that picking $\delta$ very close to $t$ (for example $\delta= t^{0.9}$), the results are very close to the Path-Greedy spanner, and the number of the performed shortest paths queries is still small. 
Moreover, the experiments show that the number of shortest path queries is linear while selecting $\delta =t$ and obtaining the $\delta$-Greedy spanner identical to the Path-Greedy $t$-spanner.

The experiments presented in this paper were performed on set of points placed independently at random in a unit square.
However, we conjecture that the $\delta$-Greedy algorithm computes a $t$-spanner in expected $O(n \log n) $ time 
for almost all realistic inputs, that is, the $\delta$-Greedy algorithm computes a $t$-spanner in expected $O(n \log n) $ time  
for point sets that are not deliberately hand-made to cause  a higher number of shortest path queries.

\section{Acknowledgments}\label{sec:ack}
We would like to thank Rachel Saban for implementing the algorithms.




\begin{thebibliography}{10}

\bibitem{AlewijnseBBB15}
S.~P.~A. Alewijnse, Q.~W. Bouts, A.~P. ten Brink, and K.~Buchin.
\newblock Computing the greedy spanner in linear space.
\newblock {\em Algorithmica}, 73(3):589--606, 2015.

\bibitem{Alewijnse2016}
S.~P.~A. Alewijnse, Q.~W. Bouts, Alex~P. ten Brink, and K.~Buchin.
\newblock Distribution-sensitive construction of the greedy spanner.
\newblock {\em Algorithmica}, pages 1--23, 2016.

\bibitem{AryaMS94}
S.~Arya, D.~M. Mount, and M.~H.~M. Smid.
\newblock Randomized and deterministic algorithms for geometric spanners of
  small diameter.
\newblock In {\em FOCS}, pages 703--712, 1994.

\bibitem{BCFMS08}
Prosenjit Bose, Paz Carmi, Mohammad Farshi, Anil Maheshwari, and Michiel H.~M.
  Smid.
\newblock Computing the greedy spanner in near-quadratic time.
\newblock In {\em SWAT}, pages 390--401, 2008.

\bibitem{CallahanK92}
P.~B. Callahan and S.~R. Kosaraju.
\newblock A decomposition of multi-dimensional point-sets with applications to
  k-nearest-neighbors and n-body potential fields.
\newblock In {\em STOC}, pages 546--556, 1992.

\bibitem{Callahan93}
Paul~B. Callahan.
\newblock Optimal parallel all-nearest-neighbors using the well-separated pair
  decomposition.
\newblock In {\em FOCS}, pages 332--340, 1993.

\bibitem{Chandra}
B.~Chandra, G.~Das, G.~Narasimhan, and J.~Soares.
\newblock New sparseness results on graph spanners.
\newblock {\em Int. J. Comp. Geom. and Applic.}, 5:125--144, 1995.

\bibitem{Chandra94}
Barun Chandra.
\newblock Constructing sparse spanners for most graphs in higher dimensions.
\newblock {\em Inf. Process. Lett.}, 51(6):289--294, 1994.

\bibitem{Clarkson87}
Kenneth~L. Clarkson.
\newblock Approximation algorithms for shortest path motion planning.
\newblock In {\em STOC}, pages 56--65, 1987.

\bibitem{DasHN93}
G.~Das, P.~J. Heffernan, and G.~Narasimhan.
\newblock Optimally sparse spanners in 3-dimensional {E}uclidean space.
\newblock In {\em SoCG}, pages 53--62, 1993.

\bibitem{Das1}
G.~Das and G.~Narasimhan.
\newblock A fast algorithm for constructing sparse {E}uclidean spanners.
\newblock {\em Int. J. Comp. Geom. and Applic.}, 7(4):297--315, 1997.

\bibitem{DBLP97}
Gautam Das and Giri Narasimhan.
\newblock A fast algorithm for constructing sparse {E}uclidean spanners.
\newblock {\em Int. J. Comp. Geom. and Applic.}, 7(4):297--315, 1997.

\bibitem{DBLP07}
M.~Farshi and J.~Gudmundsson.
\newblock Experimental study of geometric t-spanners: {A} running time
  comparison.
\newblock In {\em WEA}, pages 270--284, 2007.

\bibitem{FarshiG05}
Mohammad Farshi and Joachim Gudmundsson.
\newblock Experimental study of geometric \emph{t}-spanners.
\newblock In {\em ESA}, pages 556--567, 2005.

\bibitem{FarshiG09}
Mohammad Farshi and Joachim Gudmundsson.
\newblock Experimental study of geometric \emph{t}-spanners.
\newblock {\em {ACM} Journal of Experimental Algorithmics}, 14, 2009.

\bibitem{GudmundssonLN02}
Joachim Gudmundsson, Christos Levcopoulos, and Giri Narasimhan.
\newblock Fast greedy algorithms for constructing sparse geometric spanners.
\newblock {\em {SIAM} J. Comput.}, 31(5):1479--1500, 2002.

\bibitem{Kei88}
J.~Mark Keil.
\newblock Approximating the complete {E}uclidean graph.
\newblock In {\em SWAT}, pages 208--213, 1988.

\bibitem{DBLP02}
Christos Levcopoulos, Giri Narasimhan, and Michiel H.~M. Smid.
\newblock Improved algorithms for constructing fault-tolerant spanners.
\newblock {\em Algorithmica}, 32(1):144--156, 2002.

\bibitem{GiriSmid07}
Giri Narasimhan and Michiel Smid.
\newblock {\em Geometric Spanner Networks}.
\newblock Cambridge University Press, New York, NY, USA, 2007.

\bibitem{Soares1994}
Jos{\'e} Soares.
\newblock Approximating {E}uclidean distances by small degree graphs.
\newblock {\em Discrete {\&} Computational Geometry}, 11(2):213--233, 1994.

\end{thebibliography}
\end{document}